\documentclass[preprint,12pt]{elsarticle}
\usepackage{etoolbox}
\makeatletter
\def\ps@pprintTitle{%
	\let\@oddhead\@empty
	\let\@evenhead\@empty
	\def\@oddfoot{\reset@font\hfil\thepage\hfil}
	\let\@evenfoot\@oddfoot
}
\makeatother
\usepackage{graphics}
\usepackage[table,xcdraw]{xcolor}
\usepackage{amsmath, amsthm, amssymb}
\usepackage{natbib}
\usepackage[colorlinks=true]{hyperref}
\usepackage{lscape}
\usepackage[utf8]{inputenc} 
\usepackage{amsthm}

\usepackage{enumerate}
\usepackage{mathrsfs}
\usepackage{setspace}
\usepackage{multirow}
\usepackage{graphicx}
\usepackage{amsfonts}
\usepackage{verbatim}
\usepackage{amssymb}
\usepackage{amsthm}
\usepackage{diagbox}
\usepackage{booktabs}
\usepackage{array}
\usepackage{multirow,bigstrut,bigdelim}
\newtheorem{thm}{Theorem}[section]
\newtheorem{lem}{Lemma}[section]
\newtheorem{cor}{Corollary}[section]
\theoremstyle{plain}
\newtheorem{defi}{Definition}[section]
\newtheorem{ex}{Example}[section]

\newtheorem{ass}{Assumptions}[section]

\theoremstyle{remark}
\newtheorem{rem}{Remark}[section]

\numberwithin{equation}{section}
\usepackage[mathscr]{euscript}
\usepackage{geometry} 
\usepackage{graphicx,lscape} 
\usepackage{booktabs}
\usepackage{array}
\usepackage{paralist} 
\usepackage{verbatim}
\usepackage{subfig} 
\usepackage{orcidlink}

\usepackage[mathscr]{euscript}
\usepackage{geometry} 
\usepackage{graphicx,lscape} 
\usepackage{booktabs} 
\usepackage{bigstrut}
\usepackage{array} 
\usepackage{paralist} 
\usepackage{verbatim} 
\usepackage{subfig} 
\biboptions{comma,round,authoryear}

\newcommand{\E}{\mathbb{E}}

\newcommand{\R}{\mathbb{R}}

\biboptions{comma,round,authoryear}
   
\begin{document}
\begin{frontmatter}
	\title{\textbf{A New Bivariate Excess Wealth Transform Based on Conditional Random Variables}}
	
	\author{S. M. Sunoj \orcidlink{0000-0002-6227-1506}, Anjali Babu \orcidlink{0009-0000-3699-9942}\corref{cor1}}
	\ead{smsunoj@cusat.ac.in, anjalirose@cusat.ac.in}
	
	\cortext[cor1]{Corresponding author}
	
	\address{Department of Statistics, Cochin University of Science and Technology\\Cochin 682 022, Kerala, India}
	\begin{abstract}
		In this study, we introduce a new Bivariate Excess Wealth Transform (BEWT) to study excess wealth characteristics in bivariate settings. The proposed transform is constructed using conditional quantile functions and their corresponding conditional quantile densities. We establish fundamental properties of the proposed transform and derive some characteristic relationship for some well known bivariate lifetime models. Nonparametric weighted empirical estimation of BEWT has been proposed; theoretical and numerical validation of the method has also been considered in the study, along with a real data analysis based on the Eurosystem Household Finance and Consumption Survey.
	\end{abstract}
	
	\begin{keyword}
		Excess Wealth Transform \sep Quantile Function \sep Conditional quantile function \sep Weighted empirical estimation.
	\end{keyword}
\end{frontmatter}
\section{Introduction}

The comparison of variability and inequality among probability
distributions is an important problem in statistics, probability, and
their applications. Although classical measures such as variance
and standard deviation provide convenient numerical summaries of
dispersion, they may not adequately describe differences in the
distributional tails. This has motivated the development of
distribution-based measures and stochastic orders that provide a more
detailed characterization of variability and tail behaviour.  Among such measures, the excess wealth transform (EWT) is an important
distributional tool for studying upper-tail variability. The concept was
introduced independently by \cite{fernandez1998partial}, who referred
to it as the excess wealth transform, and by \cite{shaked1998two}, who considered the corresponding right-spread transform. The excess wealth transform has subsequently been studied in connection
with stochastic orders, variability comparisons, reliability theory,
actuarial science, and economic applications. In particular, its
interpretation in terms of the wealth held above a specified quantile
makes it naturally relevant to the study of income and wealth inequality.
For a continuous non-negative random variable $X$ with a cumulative distribution function (CDF), $F(x)$ and survival function (SF) $\bar{F}(x) = 1 - F(x)$ having finite mean, the excess wealth transform (EWT) is given by
\begin{equation}\label{1.1}
	W(u)
	=
	\int_{Q(u)}^{\infty}
	\bar{F}(x)
	dx,
	\qquad 0\le u \le1,
\end{equation}
where $Q(u) = \inf [x | F(x) \geq u]$ is the quantile function of $X$. Although a good amount of literature on EWT is on ordering of distributions, the interpretation of (1.1) in different contexts makes it a versatile tool in different disciplines. For example, in economics it is the additional wealth on top of the $u$th quantile of the richest $100(1 - u)\%$ of individuals in a population, while in finance it represents the expected
shortfall of the portfolio with loss $X$ and solvency capital requirement $Q(u)$, and for the re-insurer it is the expected claim above the retention level (payment made by reinsurer on any claim in excess of this level). \cite{kochar2007excess}, \cite{belzunce2015some} and \cite{belzunce2016sufficient} investigated the role of different stochastic ordering based on EWT in comparing random variables, particularly in the contexts of generalized order statistics and sample spacings.  The relevance of EWT to reliability theory has also been explored extensively through its connections with different ageing classes and the comparison of lifetime distributions. Notable contributions in this direction include those of \cite{kochar1997connections}, \cite{alharbi2005new}, \cite{kayid2009some} and \cite{kochar2013excess}, highlighting the usefulness of EWT-based concepts in assessing reliability and characterizing ageing properties.  The EWT has also found applications in finance and risk analysis, \cite{kochar2002total}, \cite{bertsimas2004shortfall} and \cite{sordo2009comparing} examined their relevance to concepts such as tail variability, value-at-risk, and expected shortfall.  Related applications in insurance and reinsurance have also been discussed by \cite{denuit1999lorenz}.

Several multivariate extensions of distributional transforms and
stochastic orders have been considered in the literature. In particular, \cite{fernandez2011new} have introduced a multivariate extension of the excess wealth function and its application in comparing multivariate distributions. For a bivariate vector $(X_1, X_2)$
in $R^2$ with joint probability density function (PDF) $f(x_1, x_2)$ and SF, $\bar{F}(x_1, x_2)$, and marginal
SFs $\bar{F}_{X_i}(x_i), i = 1, 2$, \cite{fernandez2011new} defined the bivariate EWT (BEWT) as

\begin{equation}
	W(u_1,u_2)
	=
	\int_{Q_{X_1}(u_1)}^{\infty}
	\bar{F}_{X_1}(t)
	S_{X_2|X_1 = t_1}(u_2)
	dt,
	\qquad 0\le u_1,u_2<1,
\end{equation}
where 
\[
S_{X_2|X_1 = t_1}(u_2) = \int_{Q_{21}^*(u_1,u_2)}^{\infty}
\bar{F}_{X_2|X_1=x_1}(t_2)dt_2
\]
with $Q_{21}^*(u_1,u_2) = \inf\left\{x_2|S\left(x_2| Q_{X_1}(u_1)\right) = u_2\right\}$, $S(x_2|x_1)$ is the conditional survival function and $Q_{X_1}(u_1)$ is the quantile function of $X_1$.

Despite the usefulness of the univariate EWT, many economic and
financial phenomena are inherently multivariate. Wealth and income, for example, are often accompanied by other financial characteristics such
as financial assets, liabilities, or household income components.
Moreover, economic variables may exhibit substantial dependence, so
that studying their marginal distributions separately may fail to capture important features of their joint behaviour. A natural
extension of the EWT should therefore incorporate information from the joint distribution while retaining the upper-tail interpretation of the univariate transform.
While extending unvariate definitions to higher dimensions, it can be accomplished
in different ways depending on the manner in which the joint distribution is characterized.
Although excess wealth transforms have been extensively studied in the univariate setting and multivariate extensions of related distributional transforms have been developed, a conditional quantile-based vector formulation specifically designed to measure bivariate upper-tail excess wealth has not been systematically developed. Motivated by this fact, in the present study we consider two conditional distributions to derive the proposed bivariate excess wealth transform.  An important feature of the proposed construction is that it preserves the upper-tail nature of the univariate EWT while incorporating conditional information. The vector formulation also allows the two
directions of conditional upper-tail behaviour to be studied
simultaneously.

We have organized the paper into five sections. After the introduction, in Section 2 the definition of the proposed BEWT and the formulas to evaluate it
are given with illustrations. Some fundamental properties of BEWT are presented in Section 3. In Section 4, we propose some characterizations for BEWT based on some important bivariate lifetime models, and in Section 5 we propose a weighted empirical distribution function based method of estimation of the BEWT. Section 6 is devoted for the validation of the proposed method using a simulation study.  Finally, in Section 7, the usefulness of the proposed measure is illustrated using a real data analysis.

\section{Bivariate excess wealth transform}
Let $(X_1,X_2)$ be a non-negative absolutely continuous random vector with joint distribution function $F(x_1, x_2)$ and marginal distribution functions $F_{X_1}(x_1)$ and $F_{X_2} (x_2)$. Further assume that $F(x_1, x_2)$ is strictly increasing in each of the arguments $x_1$ and $x_2$ and $E(X_1 X_2) < \infty$. With reference to the marginals $F_i (x_i)$ and the conditional distribution functions of $X_i$ given $X_j > x_j$ denoted by
\begin{equation}\label{conditional-survival}
	F_{ij}\left(x_i | x_j\right) = F_{ij}(x_i|X_j > x_j) = P[X_i \leq x_i | X_j > x_j], \; i,j = 1,2; i \neq j
\end{equation}
the quantile functions are defined by
\begin{equation}\label{2.1}
	Q_i(u_i) =  \inf\left\{x_i | F_i(x_i)\geq u_i\right\},
	\quad 0\leq u_i\leq 1
\end{equation}
\begin{equation}\label{2.2}
	Q_{12}(u_1,u_2) =  \inf\left\{x_1 | F_{12}(x_1 | Q_2(u_2))\geq u_1\right\}
\end{equation}
and
\begin{equation}\label{2.3}
	Q_{21}(u_1,u_2) =  \inf\left\{x_2 | F_{21}(x_2 | Q_1(u_1))\geq u_2\right\},
\end{equation}
such that the conditional quantile density functions $q_{12}(u_1, u_2) = 	\frac{\partial}{\partial u_1}	Q_{12}(u_1, u_2)$ and $q_{21}(u_1, u_2) = \frac{\partial}{\partial u_2}Q_{21}(u_1, u_2)$ exist. The pair $(Q_1 (u_1), Q_{21} (u_1, u_2))$ is called the bivariate quantile function of $(X_1,X_2)$. The pair $(Q_{12} (u_1, u_2), Q_{2} (u_2))$ is also an equivalent definition to bivariate quantile function. We refer to \cite{vineshkumar2019bivariate} and \cite{unnikrishnan2023properties} for the definition, properties, bivariate quantile models based on this framework, and their applications.

Unlike the definition of a bivariate quantile function as a pair of marginal quantile function and conditional quantile functions proposed by \cite{vineshkumar2019bivariate}, the bivariate probability structure can also be determined from a pair of conditional quantile measures. Based on this idea, recently, \cite{vineshkumar2025vector} proposed a vector-valued definition of the bivariate TTT transform and showed that it uniquely determines the joint survival function or distribution function, and studied various properties of their proposed measure.  Motivated by this, we propose a straightforward extension of EWT in \eqref{1.1} to the bivariate case as a vector-valued form, defined as follows. 
\begin{defi}\label{def21}
	Let $(X_1, X_2)$ be a non-negative, absolutely continuous random vector with joint density \(f(x_1,x_2)\) and marginal distribution functions \(F_1\) and \(F_2\). Assume that \(F_1\) and \(F_2\) are continuous and strictly increasing on their respective supports and that $ E(X_1)<\infty, \; E(X_2)<\infty$.  For \(u_1,u_2\in(0,1)\), let the marginal and conditional quantile functions are defined as in \eqref{2.1}, \eqref{2.2} and \eqref{2.3}. Then the \emph{Bivariate Excess Wealth Transform} (BEWT) of $(X_1, X_2)$ is defined as a pair
	\[		
	W(u_1,u_2) = \left(W_{12}(u_1,u_2), W_{21}(u_1,u_2)\right)
	\]
	where
	\begin{equation}\label{11}
		W_{12}(u_1,u_2) = \int_{Q_{12} (u_1, u_2)}^{\infty}{\bar{F}_{12}\left(x_1 | Q_2 (u_2)\right)dx_1}
	\end{equation}
	and 
	\begin{equation}\label{12}
		W_{12}(u_1,u_2) = \int_{Q_{21} (u_1, u_2)}^{\infty}{\bar{F}_{21}\left(x_2 | Q_1 (u_1)\right)dx_2},
	\end{equation}
	provided both the integrals are finite and $\bar{F}_{12}\left(x_1 | Q_2 (u_2)\right) = 1 - F_{12}\left(x_1|X_2 > Q_2(u_2)\right)$ and $\bar{F}_{21}\left(x_2 | Q_1 (u_1)\right) = 1 - F_{21}\left(x_2|X_1 > Q_1(u_1)\right)$ denote the conditional survival functions of $\left(X_1|X_2 > Q_2(u_2)\right)$ and $\left(X_2|X_1 > Q_1(u_1)\right)$ respectively.
\end{defi}
\begin{ex}\label{ex21}
	Consider the Gumbel bivariate exponential distribution with joint survival function
	\begin{equation}\label{26}
		\bar {F}(x_1, x_2) = \exp\{-\lambda_1x_1-\lambda_2x_2-\theta x_1x_2\}, \; x_1, x_2>0, \; \lambda_1, \lambda_2 > 0, \; 0 < \theta < \lambda_1 \lambda_2.
	\end{equation}
	with marginals $\bar{F}_i (x_i) = \exp \{-\lambda_i x_i\}$, and quantile functions $Q_i (u_i) = -\frac{1}{\lambda_i}\log (1 - u_i), \; i = 1,2$.  The conditional survival function is given by $\bar{F}_{ij}(x_i|x_j) = \exp\{-(\lambda_i + \theta x_j) x_i\}, \; i, j = 1, 2; \; i \neq j$, and subsequently 
	\[
	\bar{F}_{ij} \left(x_i | Q_j (u_j)\right) = \exp \Big\{-\left(\lambda_i - \frac{\theta}{\lambda_j} \log (1 - u_j)\right)x_i\Big\}, \; i, j = 1, 2; \; i \neq j.
	\]
	Setting $x_i = Q_{ij}(u_1, u_2)$, using the above conditional probability and applying the inverse transform,
	\[
	Q_{ij}(u_1, u_2) = \frac{-\lambda_j \log (1 - u_i)}{\lambda_1 \lambda_2 - \theta \log (1 - u_j)}, \; i, j = 1, 2; \; i \neq j.
	\]
	Using \eqref{11} and \eqref{12}, the BEWT of $(X_1, X_2)$ is obtained as
	\begin{equation}\label{13}
		\left(W_{12}(u_1,u_2), W_{21}(u_1,u_2)\right) = \left(\frac{\lambda_2 (1 - u_1)}{\lambda_1 \lambda_2 - \theta \log (1 - u_2)}, \frac{\lambda_1 (1 - u_2)}{\lambda_1 \lambda_2 - \theta \log (1 - u_1)}\right).
	\end{equation}
	When $\theta = 0$, the independence case in which Gumbel's bivariate exponential is a product of univariate exponentials, the BEWT becomes
	\[
	\left(W_{12}(u_1,u_2), W_{21}(u_1,u_2)\right) = \left(\frac{1 - u_1}{\lambda_1}, \frac{1 - u_2}{\lambda_2}\right).
	\]
	It can also observed from \eqref{13} that when there is a dependence ($\theta > 0$), $\frac{\partial W_{12}}{\partial u_2} < 0$ and $\frac{\partial W_{21}}{\partial u_1} < 0$.  That is, as the conditioning threshold of $ X_2 (X_1)$ increases, the BEWT of $X_1 (X_2)$ decreases.
\end{ex}
One can also define BEWT in Definition \ref{def21} to an alternative way by setting $x_i = Q_{ij}(u_1, u_2)$ and $\bar{F}_{ij} \left(x_i | Q_j (u_j)\right) = u_i, \; i, j = 1, 2; \; i \neq j$.  The BEWT is the pair that can be expressed in terms of conditional quantile density functions $q_{12}(u_1, u_2)$ and $q_{21}(u_1, u_2)$ as follows.

\begin{defi}\label{def22}
	Let \((X_1,X_2)\) be a non-negative absolutely continuous random vector with joint distribution function $F(x_1, x_2)$ and marginal distribution functions $F_{X_1}(x_1)$ and $F_{X_2} (x_2)$, and joint density $f(x_1, x_2)$.  Assume that \(F\) is strictly increasing in each argument and that \(E(X_1)<\infty\) and \(E(X_2)<\infty\). For every \(u_1,u_2\in[0,1)\), suppose that the conditional distributions \(X_1 | X_2>Q_2(u_2)\) and \(X_2 | X_1 > Q_1(u_1)\) are absolutely continuous and possess absolutely continuous quantile functions with quantile densities \(q_{12}(p_1 | u_2)\) and \(q_{21}(p_2 | u_1)\), respectively. Further, assume that \(P(X_1 > Q_1(u_1)) > 0\) and \(P(X_2 > Q_2(u_2)) > 0\).  Then the \emph{Bivariate Excess Wealth Transform} (BEWT) of $(X_1, X_2)$ is defined as a pair
	\[		
	W(u_1,u_2) =
	\left(W_{12}(u_1,u_2),
	W_{21}(u_1,u_2)
	\right)
	\]
	where
	\begin{equation}\label{2.4}
		W_{12}(u_1,u_2)
		=
		\int_{u_1}^{1}
		(1-p_1)\,
		q_{12}(p_1, u_2)\,
		dp_1,
		\qquad 0\le u_1, u_2<1,
	\end{equation}
	and	
	\begin{equation}\label{2.5}
		W_{21}(u_1, u_2)
		=
		\int_{u_2}^{1}
		(1-p_2)\,
		q_{21}(u_1, p_2)\,
		dp_2,
		\qquad 0\le u_1,u_2<1,
	\end{equation}
	provided both the integrals are finite.
\end{defi}
The first component $W_{12}(u_1,u_2)$ measures the upper-tail excess associated with the first variable at the level $u_1$, while incorporating conditional information from the second variable at the level $u_2$.  Similarly, $W_{21}(u_1,u_2)$ describes the corresponding upper-tail excess of the second variable conditioned on the first. Thus, unlike the univariate EWT, the proposed transform simultaneously reflects upper-tail behaviour and the dependence between the two variables.

The conditional formulation is particularly useful in economic applications. For example, when $(X_1,X_2)$ represents two related economic characteristics, $W_{12}(u_1,u_2)$ can be interpreted as the additional amount of the first characteristic above its $u_1$th quantile, taking into account the position of the second characteristic at the level $u_2$. The second component has an analogous interpretation, with the roles of the two variables reversed. This
provides a flexible framework for studying conditional inequality and upper-tail concentration in bivariate economic data.
\begin{ex}\label{ex22}
	Consider the bivariate Lomax (Pareto II) distribution with joint survival function \citep{mardia1962multivariate}
	\begin{equation}\label{Pareto II}
		\bar{F}(x_1,x_2)
		=
		\left(
		1+\frac{x_1}{\beta_1}
		+\frac{x_2}{\beta_2}
		\right)^{-\alpha},
		\qquad x_1>\beta_1>0,x_2>\beta_2>0,
	\end{equation}
	with Pareto type II marginals $\bar{F}_i(x_i)
	=
	\left(1+\frac{x_i}{\beta_i}\right)^{-\alpha},
	x_i > \beta _i > 0, i = 1,2$.  The marginal quantile function is
	\[
	Q_i(u_i)
	=
	\beta_i
	\left[
	(1-u_i)^{-1/\alpha}-1
	\right],
	\qquad 0<u_i<1, \; i = 1, 2.
	\]
	Then the conditional survival function of $X_1$ given $X_2=x_2$ is, 
	\[
	\bar{F}_{12}(x_1 | x_2)
	=
	\left[
	1+
	\frac{x_1}
	{\beta_1(1+x_2/\beta_2)}
	\right]^{-\alpha-1}.
	\]
	It follows that the conditional quantile function of $X_1$ given $X_2=x_2$ is
		\[
	Q_{12}(u_1, u_2)
	=
	\beta_1
	(1-u_2)^{-1/\alpha}
	\left[
	(1-u_1)^{-1/(\alpha+1)}-1
	\right]
	\]
	and the conditional quantile density is obtained as
	\[
	q_{12}(u_1, u_2)
	=
	\frac{\beta_1}{\alpha+1}
	(1-u_2)^{-1/\alpha}
	(1-u_1)^{-(\alpha+2)/(\alpha+1)}.
	\]
	Using \eqref{2.4} we obtain
	\[
	W_{12}(u_1,u_2)
	=
	\frac{\beta_1}{\alpha}
	(1-u_2)^{-1/\alpha}
	(1-u_1)^{\alpha/(\alpha+1)}.
	\]
	Similarly, by interchanging the roles of $X_1$ and $X_2$, we obtain
	\[
	W_{21}(u_1,u_2)
	=
	\frac{\beta_2}{\alpha}
	(1-u_1)^{-1/\alpha}
	(1-u_2)^{\alpha/(\alpha+1)}.
	\]
	Figure \ref{fig1} provides the geometry of the pair of BEWT for bivariate Pareto distribution.
	\begin{figure}[ht!]
		\begin{center}
			\subfloat[$W_{12}(u_1, u_2)$]{\includegraphics[width=0.49\textwidth]{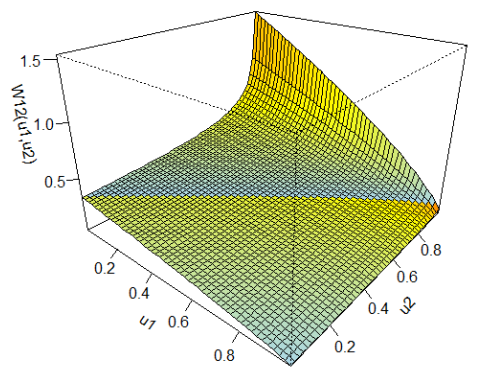}}
			\subfloat[$W_{21}(u_1, u_2)$]{\includegraphics[width=0.49\textwidth]{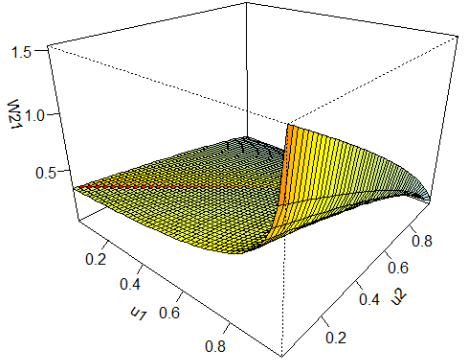}}
		\end{center}
		\caption{BEWT for Bivariate Pareto II}%
	\end{figure}\label{fig1}
\end{ex}

\begin{ex}
	Consider the following conditional quantile functions, given by
	\[
	Q_{12}(u_1, u_2) = \lambda_1 \left[(1 - u_1)^{-1/\alpha} - 1\right](1 - u_2)^{-\beta}
	\]
	and
	\[
	Q_{21}(u_1, u_2) = \lambda_2 \left[(1 - u_2)^{-1/\alpha} - 1\right](1 - u_1)^{-\beta}
	\]
	where $\lambda_1, \lambda_2 > 0, \alpha > 1$ and $\beta \geq 0$.  Note that the above conditional quantile densities do not possess a closed form DF or SF, that makes the significance of the defintion of BEWT in \eqref{2.4} and \eqref{2.5}, which provide the BEWT as
	\[
	W_{12}(u_1, u_2) = \frac{\lambda_1}{\alpha - 1} (1 - u_1)^{1 - 1/\alpha}(1 - u_2)^{\beta}
	\]
	and
	\[
	W_{21}(u_1, u_2) = \frac{\lambda_2}{\alpha - 1} (1 - u_2)^{1 - 1/\alpha}(1 - u_1)^{\beta}.
	\]
\end{ex}

The BEWTs in \eqref{2.4} and \eqref{2.5} can also be written as
\begin{equation}\label{2.10}
	W_{12}(u_1, u_2) = \int_{u_1}^{1}{\left(Q_{12}(p_1, u_2) - Q_{12}(u_1, u_2)\right)dp_1} = \int_{u_1}^{1}{Q_{12}(p_1, u_2)dp_1} - (1 - u_1)Q_{12}(u_1, u_2),
\end{equation}
and
\begin{equation}\label{2.11}
	W_{21}(u_1, u_2) = \int_{u_2}^{1}{\left(Q_{21}(u_1, p_2) - Q_{21}(u_1, u_2)\right)dp_2} = \int_{u_2}^{1}{Q_{21}(u_1, p_2)dp_2} - (1 - u_2)Q_{21}(u_1, u_2)
\end{equation}
Equations \eqref{2.10} and \eqref{2.11} are equivalent to 
\begin{equation}\label{2.12}
	W_{12}(u_1,u_2) = \E \left[\left\{Q_{12}(X_1, u_2) - Q_{12}(u_1, u_2)\right\}_{+} |X_1 > u_1)\right],
\end{equation}
and
\begin{equation}\label{2.13}
	W_{21}(u_1,u_2) = \E \left[\left\{Q_{21}(u_1, X_2) - Q_{21}(u_1, u_2)\right\}_{+} |X_2 > u_2\right],
\end{equation}
where $a_+=\max(a, 0)$, denote respectively the conditonal stop-loss transforms, obtained as an extension of the univariate stop-loss transform given in  \cite{nair2013quantile}.
\begin{rem}
	The representation in \eqref{2.12} and \eqref{2.13} are particularly useful in empirical applications. It avoids numerical differentiation of an empirical conditional quantile function and therefore provides a more stable estimator of the BEWT.
\end{rem}
\subsection{Normalized Bivariate Excess Wealth Transform}

The values of the bivariate excess wealth transforms $W_{12}$ and $W_{21}$ depend on the scales of the underlying random variables.  Since the two components may be measured on substantially different
scales, a scale-free version is useful for comparative analysis.
Accordingly, we introduce normalized forms of the BEWT, obtained by taking the ratio of BEWTs to its corresponding marginal EWTs, given respectively by 
\begin{equation}\label{216}
	\widetilde{W}_{12}(u_1,u_2)
	=
	\frac{W_{12}(u_1,u_2)}
	{W_1(u_1)}
\end{equation}
and
\begin{equation}\label{217}
	\widetilde{W}_{21}(u_1,u_2)
	=
	\frac{W_{21}(u_1,u_2)}
	{W_2(u_2)},
\end{equation}
provided $W_1(u_1)>0$ and $W_2(u_2)>0$.

The normalized BEWT has a direct interpretation in terms of dependence. When $X_1$ and $X_2$ are independent, we have $F_{ij}\left(x_i | Q_j (u_j)\right) = F_i (x_i), \; i, j = 1, 2;  \; i \neq j$ and $Q_{ij}(u_1, u_2) = Q_i (u_i), \; i, j = 1, 2;  \; i \neq j$, using \eqref{11} and \eqref{12}, we get $W_{ij} (u_1, u_2) = W_i (u_i),  \; i, j = 1, 2;  \; i \neq j$, which in turn gives \eqref{216} and \eqref{217} to 
\[
\widetilde{W}_{12}(u_1,u_2)
=
\widetilde{W}_{21}(u_1,u_2)
=1.
\]
Thus, departures of $\widetilde{W}_{12}$ and $\widetilde{W}_{21}$
from unity provide a natural way of quantifying the effect of
conditioning. This
property makes the normalized BEWT particularly useful
for studying dependence between the two components.

\section{Fundamental Properties of the BEWT}

\begin{itemize}
	\item \textbf{Non-negativity}
	For every $u_1,u_2\in(0,1)$, $W_{12}(u_1,u_2)\geq 0$ and $W_{21}(u_1,u_2)\geq 0$.
	\begin{proof}
		Since $Q_{12}(u_1, u_2)$ and $Q_{21}(u_1, u_2)$ are non-decreasing
		quantile functions, their derivatives, the quantile density functions satisfy $q_{12}(u_1, u_2)\geq 0$ and $q_{21}(u_1, u_2)\geq 0$.  Moreover, for $p\in[u_1,1]$ and $p\in[u_2,1]$, respectively, gives $1-p_1\geq 0$, and hence $(1-p_1)q_{12}(p_1, u_2) \geq 0$ and $(1-p_2)q_{21}(u_1, p_2) \geq 0$, proves the result.  Consequently,
		\[
		W_{12}(u_1,u_2)
		=
		\int_{u_1}^{1}
		(1-p_1)q_{12}(p_1, u_2)\,dp
		\geq0
		\]
		and
		\[
		W_{21}(u_1,u_2)
		=
		\int_{u_2}^{1}
		(1-p_2)q_{21}(u_1, p_2)\,dp
		\geq0.
		\]
	\end{proof}
	\item \textbf{Boundary conditions}
	For every fixed $u_2\in(0,1)$ and $u_1\in(0,1)$;
	\begin{align*}
		\text{(a)}\quad & \lim_{u_1\to1}W_{12}(u_1,u_2) =0, \quad \lim_{u_2\to1}W_{21}(u_1, u_2) =0\\
		\text{(b)}\quad & \lim_{u_2\to 0}W_{12}(u_1,u_2) = W_{X_1}(u_1), \quad \lim_{u_1\to 0}W_{21}(u_1,u_2) = W_{X_2}(u_2)
	\end{align*}
	\begin{proof}
		The proof (a) directly follows from \eqref{2.4} and \eqref{2.5}.  Since $\lim_{u_2\to 0}Q_{12}(u_1,u_2)= Q_1(u_1)$ and $\lim_{u_1\to 0}Q_{21}(u_1,u_2)= Q_2(u_2)$, (b) follows.
	\end{proof}
	
	\item \textbf{Monotonicity:}
	For fixed $u_2\in(0,1)$, $W_{12}(u_1,u_2)$ is non-increasing in $u_1$.
	Similarly, for fixed $u_1\in(0,1)$, $W_{21}(u_1,u_2)$ is non-increasing
	in $u_2$.
	\begin{proof}
		By the fundamental theorem of calculus, we have
		\[
		\frac{\partial}{\partial u_1}
		W_{12}(u_1, u_2)
		=
		-(1-u_1)q_{12}(u_1, u_2).
		\]
		Since $1-u_1\geq 0$ and $q_{12}(u_1, u_2)\geq 0$, we obtain $\frac{\partial}{\partial u_1}W_{12}(u_1, u_2) \leq 0$.  Therefore, $W_{12}(u_1,u_2)$ is non-increasing in $u_1$. In a similar manner, $W_{21}(u_1,u_2)$ is non-increasing in $u_2$.
	\end{proof}
	\item \textbf{Continuity:}
	Suppose that $q_{12}(u_1, u_2)$ is continuous in $u_2$ for almost every $u_1 \in (u_1, 1)$ and that there exists an integrable function $g_1(u_1)$ such that $(1-u_1)q_{12}(u_1, u_2)\leq g_1(u_1)$ for all $u_2$ in a neighbourhood of a fixed point. Then $W_{12}(u_1, u_2)$ is continuous in $u_2$.  Similarly, if there exists an integrable function $g_2(u_2)$ such that $(1-u_2)q_{21}(u_1, u_2)\leq g_2(u_2)$, then $W_{21}(u_1, u_2)$ is continuous in $u_1$.
	\begin{proof}
		Let $u_2^{(n)}\to u_2$. By assumption, $q_{12}(u_1, u_2^{(n)})
		\longrightarrow q_{12}(u_1, u_2)$ for almost every $p$. Moreover, $(1-u_1)q_{12}(u_1, u_2^{(n)})
		\leq g_1(u_1)$, where $g_1$ is integrable. Hence, by the dominated convergence theorem,
		\[
		\begin{aligned}
			\lim_{n\to\infty}
			W_{12}(u_1,u_2^{(n)})
			&=
			\lim_{n\to\infty}
			\int_{u_1}^{1}
			(1-p_1)q_{12}(p_1, u_2^{(n)})\,dp_1\\
			&=
			\int_{u_1}^{1}
			(1-p_1)q_{12}(p_1, u_2)\,dp_1 = W_{12}(u_1,u_2).
		\end{aligned}
		\]
		Thus, $W_{12}$ is continuous in $u_2$. The proof for $W_{21}$ is
		analogous.
	\end{proof}
	\item \textbf{Translation Invariance:}
	Let $X_1^*=X_1+a,\quad X_2^*=X_2+b$, where $a,b\in\mathbb{R}$. Then
	\[
	W_{12}^*(u_1,u_2)=W_{12}(u_1,u_2) \quad \text{and} \quad W_{21}^*(u_1,u_2)=W_{21}(u_1,u_2).
	\]
	
	\begin{proof}
		Since quantile function possess location invariant property, we have when $Y = X + a$, then $Q_Y(u) = Q_X (u) + a$.  Applying this property we have, 
		\[
		Q_{X_1^*}(u_1)=Q_{X_1}(u_2)+a \quad \text{and} \quad Q_{X_2^*}(u_2)=Q_{X_2}(u_2)+b.
		\]
		In a similar manner,
		\[
		Q_{12^*}(u_1, u_2) =
		Q_{12^*}(u_1, u_2)+a \quad \text{and} \quad Q_{21^*}(u_1, u_2) = Q_{21^*}(u_1, u_2)+b
		.
		\]
		Differentiating with respect to $u_i, \; i = 1,2$ gives
		\[
		q_{12}^*(u_1, u_2)
		=
		q_{12}(u_1, u_2) \quad \text{and} \quad q_{21}^*(u_1, u_2)
		= q_{21}(u_1, u_2),
		\]
		which proves the result.
	\end{proof}
	\item \textbf{Positive Scale Equivariance:}
	Let $X_1^*=aX_1,\; X_2^*=bX_2, \; a, b>0$.  Then
	\[
	W_{12}^*(u_1,u_2) = aW_{12}(u_1,u_2) \quad \text{and} \quad W_{21}^*(u_1,u_2)
	=
	bW_{21}(u_1,u_2).
	\]
	\begin{proof}
		The proof follows from the scale-invariant property of quantile function that if $Y = aX, \; a > 0$ then $Q_Y(u) = a Q_X (u)$.
	\end{proof}
	\item \textbf{Convexity and Concavity:}
	Suppose that $q_{12}(u_1, u_2)$ is differentiable with respect to $u_i, \; i = 1, 2$. Then $W_{ij}(u_1, u_2), \; i, j = 1, 2, \; i \neq j$ is convex (concave) in $u_i$ if
	\[
	q_{ij}(u_1, u_2) \geq (\leq)
	(1-u_i) \frac{\partial}{\partial u_i}
	q_{ij}(u_1, u_2), \; i, j = 1, 2, \; i \neq j
	\]
	\begin{proof}
		From \eqref{2.4} and \eqref{2.5}, we have 
		\[
		\frac{\partial W_{ij}(u_1, u_2)}
		{\partial u_i}
		=
		-(1-u_i)q_{ij}(u_1, u_2), \; i, j = 1, 2, \; i \neq j.
		\]
		Differentiating once more with respect to $u_i$ gives
		\[
		\frac{\partial^2 W_{ij}(u_1, u_2)}
		{\partial u_i^2}
		= q_{ij}(u_1, u_2)
		-
		(1-u_i)
		\frac{\partial}{\partial u_i}
		q_{ij}(u_1, u_2), \; i, j = 1, 2, \; i \neq j.
		\]
		Thus, the stated conditions determine the convexity or concavity of
		$W_{ij}(u_1, u_2)$. 
	\end{proof}
\end{itemize}

\section{Characterizations of BEWT}
For the non-negative random vector $(X_1, X_2)$, \cite{vineshkumar2019bivariate} have extended the notion of hazard quantile, denoted by $H(u)= h\left(Q(u)\right)= \frac{1}{(1 - u)q(u)}$, where $h(x) = \frac{f(x)}{\bar{F}(x)}$ to the bivariate case, keeping the dependence structure by defining it as a pair 
\[
\left(H_1(u_1), H_{21}(u_1, u_2)\right) \quad \text{or} \quad \left(H_{12}(u_1, u_2), H_2(u_2)\right)
\]
where 
\begin{equation}\label{conditional-hazard}
	H_i (u_i) = \frac{1}{(1 - u_i)q_i (u_i)}, \; i = 1, 2
\end{equation}
denotes the hazard quantile function of the random variable $X_i, \; i = 1, 2$, and 
\[
H_{ij}(u_1, u_2) = \frac{1}{(1 - u_i) q_{ij}(u_1, u_2)}, \; i, j = 1, 2, \; i \neq j
\]
denotes the conditional hazard quantile function corresponding to ${X_i|X_j > Q_j (u_j),} \; i, j = 1, 2, \; i \neq j$. Using the Definition \ref{def22}, the BEWT in terms of conditional hazard quantile functions, given by
\begin{equation}\label{31}
	W_{12}(u_1, u_2) = \int_{u_1}^1 {\frac{dp_1}{H_{12}(p_1, u_2)}}
\end{equation}
and
\begin{equation}\label{32}
	W_{21}(u_1, u_2) = \int_{u_2}^1 {\frac{dp_2}{H_{21}(u_1, p_2)}}.
\end{equation}
\begin{ex}
	If $(X_1, X_2)$ follows Gumbel's bivariate exponential with survival function \eqref{26}, then
	\[
	W_{ij}(u_1, u_2) = \frac{1 - u_i}{H_{ij}(u_1, u_2)}, \; i = 1,2, \; i \neq j.
	\]
\end{ex}
Another important measure useful in quantile-based lifetime analysis is the mean residual quantile function, defined by $M(u) = m\left(Q(u)\right) = \frac{1}{1- u}\int_u^1{(1 - p)q(p)dp}$, where $m(x) = E(X - x | X >x) = \frac{1}{\bar{F}(x)} \int_x^{\infty}{\bar{F}(t)dt}$ is the mean residual life function of the random variable $X - x | X > x$. Like the bivariate hazard quantile function as a pair of marginal hazard and conditional quantile functions, \cite{vineshkumar2019bivariate} have defined the mean residual quantile function as a vector,
\[
\left(M_{1}(u_1), M_{21}(u_1, u_2)\right)
\]
where $M_1 (u_1)= \int_{u_1}^1{(1 - p)q_1 (p)dp}$, the mean residual quantile function of $X_1$, and
\begin{equation}\label{conditional-mrl-2}
	M_{21}(u_1, u_2) = \frac{1}{1 - u_2}\int_{u_2}^{1}{(1 - p_2) q_{21}(u_1, p_2)dp_2} = \frac{W_{21}(u_1, u_2)}{1 - u_2}
\end{equation}
is the mean residual quantile function corresponding to $X_2|X_1 > Q_1 (u_1)$.  The bivariate mean residual quantile function can also be defined as an alternative way by defining it $\left(M_{12}(u_1, u_2), M_{2}(u_2)\right)$, where 
\begin{equation}\label{conditional-mrl-1}
	M_{12}(u_1, u_2) = \frac{1}{1 - u_1}\int_{u_1}^{1}{(1 - p_1) q_{12}(p_1, u_2)dp_1} = \frac{W_{12}(u_1, u_2)}{1 - u_1},
\end{equation}
and $M_2 (u_2)= \int_{u_2}^1{(1 - p)q_2 (p)dp}$ denote respectively the mean residual quantile functions corresponding to $X_1|X_2 > Q_2 (u_2)$ and $X_2$ (see also \citet{vineshkumar2025vector}) and therefore
\begin{equation}\label{33}
	\left(W_{12}(u_1, u_2), W_{21}(u_1, u_2)\right) = \left((1 - u_1) m_{12}(u_1, u_2), (1 - u_2) m_{21}(u_1, u_2)\right).
\end{equation}
Equation \eqref{33} provides a one-to-one relationship between BEWT and vector-valued conditional mean residual quantile function.  In other words, $m_{12}$ and $m_{21}$ characterizes $W_{12}$ and $W_{21}$.

\begin{thm}
	Let $(X_1, X_2)$ be a bivariate random vector with marginal distribution functions $F_{X_1}$ and $F_{X_2}$ that are absolutely continuous with quantile densities $q_{X_1}$ and $q_{X_2}$ existing on (0,1). Suppose that $X_1$ and $X_2$ are independent. Then $W_{12}(u_1, u_2) = W_{X_1}(u_1)$ and $W_{21}(u_1, u_2) = W_{X_2}(u_2)$.  
\end{thm}

\begin{proof}
	Since $X$ and $Y$ are independent, the events $\{X_1 \leq x_1\}$ and $\{X_2 > Q_{X_2}(u_2)\}$ are independent. Hence, $F_{12}\left(x_1 | Q_2 (u_2)\right) = F_{1}(x_1)$.  Consequently, $Q_{12}(u_1, u_2) = Q_{1}(u_1)$ and $q_{12}(u_1 | u_2) = q_{1}(u_1)$, which upon substitution in \eqref{2.4} proves the result. The other case is similar.
\end{proof}

Next, we prove that $W(u_1, u_2)$ uniquely determines the underlying joint distribution function or survival function.
\begin{thm}
	Let $(X_1, X_2)$ be absolutely continuous non-negative random vector with joint distribution function $F(x_1, x_2)$ and marginal distribution functions $F_{X_1}(x_1)$ and $F_{X_2} (x_2)$, and joint density $f(x_1, x_2)$.  Assume that \(F\) is strictly increasing in each argument and that $E(X_1)<\infty$ and $E(X_2)<\infty$. For every $u_1, u_2 \in [0,1)$, suppose that the conditional distributions $(X_1 | X_2 > Q_2(u_2))$ and $(X_2 | X_1 > Q_1(u_1))$ are absolutely continuous and possess absolutely continuous quantile functions with quantile densities $q_{12}(u_1, u_2)$ and $q_{21}(u_1, u_2)$, respectively. Also, assume that $Q_1(u_1) =0$ and $Q_2(u_2) =0$.  Then BEWT, $W(u_1, u_2)$ uniquely determines $F$, or equivalently $\bar{F}$.
	\begin{proof}
		Differentiating \eqref{2.4} with respect to $u_1$, we have
		\[
		\frac{\partial W_{12}(u_1, u_2)}{\partial u_1} = -(1 - u_1) q_{12}(u_1, u_2),
		\]
		equivalently
		\[
		q_{12}(u_1, u_2) = -\frac{1}{1 - u_1} \frac{\partial W_{12}(u_1, u_2)}{\partial u_1}
		\]
		Using $Q_{12}(0, u_2) = 0$, integration with respect to $u_1$ gives
		\begin{equation}\label{214}
			Q_{12}(u_1, u_2) = \int_{0}^{u_1}{\frac{1}{1 - p_1} \frac{\partial W_{12}(p_1, u_2)}{\partial p_1}dp_1}.
		\end{equation}
		Similarly, 
		\begin{equation}\label{215}
			Q_{21}(u_1, u_2) = \int_{0}^{u_2}{\frac{1}{1 - p_2} \frac{\partial W_{21}(u_1, p_2)}{\partial p_2}dp_2}.
		\end{equation}
		Thus, $W_{12}$ and $W_{21}$ uniquely determine the conditional quantile functions $Q_{12}$ and $Q_{21}$ respectively. 
		
		Setting $u_2 = 0$, the random variable $(X_1 | X_2 > Q_2(0))$ reduces to $(X_1 | X_2 > 0)$ which, in turn, reduces to $X_1$.  Thus, $Q_{12}(u_1, 0) = Q_1 (u_1)$ and similarly, $Q_{21}(0, u_2) = Q_2 (u_2)$ and can be obtained directly from $W_{12}$ and $W_{21}$ using \eqref{214} and \eqref{215}.   Once the marginals $Q_1$ and $Q_2$ are known, setting, say, $x_2 = Q_2(u_2)$, the conditional quantile function $Q_{12}(u_1, u_2)$ becomes quantile function of the random variable $(X_1 | X_2 > x_2)$.  Now by setting $x_1 = Q_{12} (u_1, u_2) = Q_{X_1 | X_2 > x_2} (u_1, u_2)$ and its relationship with the conditional survival function $\bar{F}_{12}(x_1| x_2)$, the joint survival function $\bar{F}(x_1, x_2)$ can be obtained.  This completes the proof.    
	\end{proof}
\end{thm}
We now prove a characterization theorem for the bivariate Gumbel's exponential distribution \eqref{26} based on the BEWT.
\begin{thm}\label{thm43}
	Let $(X_1, X_2)$ be a non-negative absolutely continuous random vector, and $W(u_1,u_2)=\bigl(W_{12}(u_1,u_2),W_{21}(u_1,u_2)\bigr)$ be its BEWT, with $u_1, u_2 \in (0,1)$.  Suppose that $c_1, c_2 > 0$ and $0 < k < c_1 c_2$. Then $(X_1, X_2)$ follows Gumbel's bivariate exponential distribution with joint survival function \eqref{26} if and only if that BEWT satisfies the relationships
	\begin{equation}\label{28}
		W_{12}(u_1, u_2) = \frac{c_2 (1 - u_1)}{c_1 c_2 - k \log (1 - u_2)}
	\end{equation}
	and
	\begin{equation}\label{29}
		W_{21}(u_1,u_2) = \frac{c_1 (1 - u_2)}{c_1 c_2 - k \log (1 - u_1)}.
	\end{equation}
	\begin{proof}
		The 'if' part is directly follows from Example \ref{ex21}.  To prove the 'only if' part, we assume that \eqref{28} and \eqref{29} holds.  Set $u_2 = 0$, from \eqref{28}, we have
		\begin{equation}\label{210}
			W_{12}(u_1, 0) = \frac{1 - u_1}{c_1}.
		\end{equation}
		But when $u_2 = 0$, the conditional threshold is the lower end point $Q_2(0) = 0$, and therefore the conditional random variable $X_1 | X_2 > 0$ reduces to a univariate random variable $X_1$, and therefore $W_{12}(u_1, 0)$ becomes the univariate excess wealth transform function of $X_1$.  That is, $W_{12}(u_1, 0) = W_1 (u_1)$.
		
		The EWT given in \eqref{1.1} can also be expressed as $W(u) = \int_{u}^1{(1 - p)q(p)dp}$, on differentiating $W(u)$ with respect to $u$, yield
		\[
		\frac{dW(u)}{du} = -(1 - u)q(u)
		\]
		and therefore from \eqref{210}, $-\frac{1}{c_1} = -(1 - u_1)q_1(u_1)$ implies that
		\[
		q_1 (u_1) = \frac{1}{c_1 (1 - u_1)}.
		\]
		Integrating with respect to $u_1$ and using $Q_1 (0) = 0$, we obtain
		\[
		Q_1 (u_1) = -\frac{1}{c_1} \log (1 - u_1)
		\]
		implies that $\bar{F}_1 (x_1) = e^{-c_1 x_1}$.  Similarly from $W_{21}(0, u_2) = \frac{1 - u_2}{c_2}$ we obtain $Q_2 (u_2) = -\frac{1}{c_2} \log (1 - u_2)$ and hence $\bar{F}_2 (x_2) = e^{-c_2 x_2}$.
		
		From the Definition \ref{def21} and using \eqref{28}, we have the fundamental relation,
		\[
		\frac{\partial W_{12}(u_1, u_2)}{\partial u_1} = -(1 - u_1) q_{12}(u_1, u_2).
		\]
		From \eqref{28}, the above relationship becomes
		\[
		q_{12}(u_1, u_2) = \frac{c_2}{(1 - u_1) \left[c_1 c_2 - k \log (1 - u_2)\right]}
		\]
		Integrating with respect to $u_1$ and using $Q_{12}(0, u_2) = 0$, yield
		\begin{equation}\label{211}
			Q_{12}(u_1, u_2) = -\frac{c_2 \log (1 - u_1)}{c_1 c_2 - k \log (1 - u_2)} = -\frac{\log (1 - u_1)}{c_1 + k Q_2 (u_2)}.
		\end{equation}
		From the relationship between conditional quantile function and conditional survival function and setting $Q_{12}(u_1, u_2) = x_1$ and $u_1 = \bar{F}_{12} \left(x_1 | Q_2 (u_2)\right)$, \eqref{211} becomes
		\[
		\bar{F}_{12} \left(x_1 | Q_2 (u_2)\right) = \exp \{-x_1\left[c_1 + k Q_2 (u_2)\right]\}.
		\]
		Setting $Q_2(u_2) = x_2$ simplifies to 
		\[
		\bar{F}_{12} \left(x_1 | x_2\right) = \exp \{-x_1\left(c_1 + k x_2\right)\},
		\]
		and multiplying with the marginal survival function $\bar{F}_2 (x_2)$, we obtain the joint pdf of Gumbel's exponential given in \eqref{26}.  This completes the proof.
	\end{proof}
\end{thm}
\begin{cor}
	Let $(X_1, X_2)$ be a nonnegative, absolutely continuous random vector with continuous, strictly increasing marginal distribution functions $F_i(x_i)$ and quantile functions $Q_i(u_i), \; i = 1, 2; \; i \neq j$.  Then $(X_1, X_2)$ follows Gumbel's bivariate exponential with joint survival function \eqref{26} if and only if the BEWT satisfy
	\[
	W_{ij}(u_1, u_2) = \frac{1 - u_i}{c_i + k Q_j (u_j)}; \; c_i, k > 0, \; i = 1,2, \; i \neq j.
	\]
\end{cor}
A more general characterization theorem in terms of the conditional hazard quantile functions can be obtained from the following lemma.
\begin{lem}\label{lem41}
	Let $(X_1, X_2)$ be a nonnegative, absolutely continuous random vector with continuous, strictly increasing marginal distribution functions $F_i(x_i)$ and quantile functions $Q_i(u_i), \; i = 1, 2; \; i \neq j$.  Then $(X_1, X_2)$ follows Gumbel's bivariate exponential with joint survival function \eqref{26} if and only if the conditional hazard quantile functions satisfy
	\begin{equation}\label{4.10}
		H_{ij}(u_1, u_2) = a_i + b \ c_i (u_j), \; i, j = 1,2; \; i \neq j,
	\end{equation}
	where $a_i, b > 0$, $c_i(u_j), \; i, j = 1, 2$ is a function of $u_j$ such that $c_1(0) = c_2 (0) = 0$. 
	\begin{proof}
		The 'if' part is directly follows from Example \ref{ex21}.
		
		To prove the 'only if' part, suppose that \eqref{4.10} holds.  Consider now $i = 1$, using the relationship between $H_{12}(u_1, u_2)$ and $q_{12}(u_1 | u_2)$, we obtain
		\[
		q_{12}(u_1, u_2) = \frac{1}{(1 - u_1) \left[a_1 + b \ c_1 (u_2)\right]}
		\]
		On integrating with respect to $u_1$, yield
		\begin{equation}\label{4.11}
			Q_{12}(u_1, u_2) = -\frac{\log (1 - u_1)}{a_1 + b \ c_1 (u_2)}
		\end{equation}
		Similarly for $i = 2$ in \eqref{4.10}, yield
		\begin{equation}\label{4.12}
			Q_{21}(u_1, u_2) = -\frac{\log (1 - u_2)}{a_2 + b \ c_2 (u_1)}.
		\end{equation}
		As $u_1 \to 0$, \eqref{4.12} becomes
		\[
		Q_2 (u_2) = -\frac{\log (1 - u_2)}{a_2}
		\]
		Setting $x_2 = Q_2(u_2)$ and $u_2 = F_2 (x_2)$, we get $\bar{F}_2 (x_2) = \exp \{-a_2 x_2\}$.  Using \eqref{4.11} and setting $x_1 = Q_{12}(u_1, u_2)$ and $u_1 = F_{12}\left(x_1| Q_2 (u_2)\right)$ with $x_2 = Q_2(u_2)$ and simplifying,
		\[
		F_{12}\left(x_1| x_2\right) = \exp \left[-a_1 x_1 - b \ c_1\left(F_2 (x_2)\right) x_1\right].
		\]
		Multiplying with the expression of $\bar{F}_2(x_2)$, we get
		\begin{equation}\label{4.13}
			\bar{F}(x_1, x_2) = \exp \left[-a_1 x_1 - a_2 x_2 - b \ c_1 \left(F_2 (x_2)\right)x_1\right].
		\end{equation}
		Starting from the second component of \eqref{4.10} for $i = 2$ and proceeding in a similar manner, we obtain        \begin{equation}\label{4.14}
			\bar{F}(x_1, x_2) = \exp \left[-a_1 x_1 - a_2 x_2 - b \ c_2 \left(F_1 (x_1)\right)x_2\right].
		\end{equation}
		Equating \eqref{4.13} and \eqref{4.14} gives,
		\[
		c_1 \left(F_2 (x_2)\right)x_1 = c_2 \left(F_1 (x_1)\right)x_2
		\]
		This is possible only if $c_1 \left(F_2 (x_2)\right) = Cx_2$ and $c_2 \left(F_1 (x_1)\right) = Cx_1$, where $C$, some positive constant independent of $x_1$ and $x_2$.  This reduces $\bar{F}(x_1, x_2)$ in \eqref{4.13} and \eqref{4.14} to the required form of Gumbel's bivariate exponential.
	\end{proof}
\end{lem}
\begin{thm}
	Under the assumptions of Lemma \ref{lem41}, $(X_1, X_2)$ follows Gumbel's bivariate exponential with joint survival function \eqref{26} if and only if the BEWT satisfy
	\[
	W_{ij}(u_1, u_2) = \frac{1 - u_i}{a_i + b \ c_i (u_j)}; \; a_i, b > 0, \; i, j = 1,2, \; i \neq j.
	\]
	\begin{proof}
		The proof directly follows from the relationship between BEWT and conditional hazard quantile function $\frac{\partial W_{ij}(u_1, u_2)}{\partial u_i} = -\frac{1}{H_{ij}(u_1, u_2)}, \; i, j = 1,2, \; i \neq j$, and applying Lemma \ref{lem41}.
	\end{proof}
\end{thm}
\begin{thm}\label{thm44}
	Let $(X_1, X_2)$ be a nonnegative, absolutely continuous random vector with continuous, strictly increasing marginal distribution functions $F_i(x_i)$ and quantile functions $Q_i(u_i), \; i = 1, 2; \; i \neq j$.  Then the BEWT satisfy the relationship,
	\begin{equation}\label{410}
		W_{ij}(u_1, u_2) = \frac{\left[1 -(1 - u_i)^{k_2}\right]}{k_1}\left(\frac{1 + \alpha_j Q_j (u_2)}{\alpha_i + \beta Q_j(u_j)}\right); \; \alpha_i > 0, \; i, j = 1,2, \; i \neq j,
	\end{equation}
	if and only $(X_1, X_2)$ follows bivariate Pareto II distribution with survival function,   \begin{eqnarray}\label{411}
		\bar{F}(x_1, x_2) = (1 + a_1 x_1 + a_2 x_2 + b x_1 x_2)^{-c},& \; x_1, x_2 > 0, \; a_1, a_2, c> 1, \nonumber\\&0 < b \leq (c + 1) a_1 a_2
	\end{eqnarray}
	according as $k_1 = k_1 - 1, k_2 = 1 - \frac{1}{k_1}, \alpha_j, \beta > 0$, or $(X_1, X_2)$ follows bivariate finite range distribution with survival function,   \begin{eqnarray}\label{412}
		\bar{F}(x_1, x_2) = (1 - p_1 x_1 - p_2 x_2 + q x_1 x_2)^d,&& 0 < x_1 < \frac{1}{p_1}, 0 < x_2 < \frac{1 - p_1 x_1}{p_2 - q x_1} \nonumber\\
		&& p_1, p_2, d > 0, 1 - d \leq \frac{q}{p_1 p_2} \leq 1,  
	\end{eqnarray}
	according as $ k_1 = k_1 + 1, k_2 = 1 + \frac{1}{k_1}, \alpha_j, \beta < 0$.
	\begin{proof}
		Assume that \eqref{411} holds.  Then we have the conditional survival function,
		\[       \bar{F}_{12}\left(x_1 | Q_2 (u_2)\right) = \left[1 + \left(\frac{a_1 + b Q_2 (u_2)}{1 + a_2 Q_2 (u_2)}\right)x_1\right]^{-c}
		\]
		and correspondingly
		\[
		Q_{12}(u_1|u_2) = \left[(1 - u_1)^{-\frac{1}{c}} - 1\right]\left(\frac{1 + a_2 Q_2 (u_2)}{a_1 + b Q_2(u_2)}\right),
		\]
		with quantile density function,
		\[
		q_{12}(u_1, u_2) = \frac{(1 - u_1)^{-\frac{1}{c} - 1}}{c}\left(\frac{1 + a_2 Q_2 (u_2)}{a_1 + b Q_2(u_2)}\right).
		\]
		Now using \eqref{2.4}, we obtain
		\[
		W_{12} (u_1, u_2) = \frac{\left[1 -(1 - u_1)^{-\frac{1}{c} + 1}\right]}{(c - 1)}\left(\frac{1 + a_2 Q_2 (u_2)}{a_1 + b Q_2(u_2)}\right)
		\]
		which takes the form of \eqref{410}.  Applying similar steps by reversing the conditioning, we obtain $W_{21}(u_1, u_2)$ with the form of \eqref{410} for $i = 2$.
		
		Conversely, assume that \eqref{410} holds, and $k_1 = k_1 - 1, k_2 = 1 - \frac{1}{k_1}, \alpha_j, \beta > 0$.  Now for $i = 1$, on differentiating with respect to $u_1$ yield
		\[
		\frac{\partial W_{12}(u_1, u_2)}{\partial u_1} = \frac{(1 - u_1)^{-\frac{1}{k_1}}}{k_1}\left(\frac{1 + \alpha_2 Q_2 (u_2)}{\alpha_1 + \beta Q_2(u_2)}\right)
		\]
		Using \eqref{214}, we obtain
		\[
		Q_{12}(u_1, u_2) = \left[(1 - u_1)^{-\frac{1}{k_1}} - 1\right]\left(\frac{1 + \alpha_2 Q_2 (u_2)}{\alpha_1 + \beta Q_2(u_2)}\right).
		\]
		Setting $Q_{12}(u_1, u_2) = x_1$ and $u_1 = F_{12}\left(x_1 | Q_2 (u_2)\right)$, and simplifying 
		\begin{equation}\label{413}       \bar{F}_{12}\left(x_1 | Q_2 (u_2)\right) = \left[1 + \left(\frac{\alpha_1 + \beta Q_2 (u_2)}{1 + \alpha_2 Q_2 (u_2)}\right)x_1\right]^{-k_1}.
		\end{equation}
		Now for $i = 2$ in \eqref{410} and $u_1 \to 0$ with $Q_1 (0) = 0$, yield
		\[
		W_2 (u_2) = \frac{\left[1 -(1 - u_2)^{1 - \frac{1}{k_1}}\right]}{\alpha_2 (k_1 - 1)}
		\]
		Differentiating with respect to $u_2$ gives
		\[
		\frac{\partial W_2 (u_2)}{\partial u_2} = \frac{(1 - u_2)^{-\frac{1}{k_1}}}{\alpha_2 k_1}.
		\]
		Using $Q_2(u_2) = \int_0^{u_2}{\frac{\frac{\partial W_2(p_2)}{\partial p_2}}{1 - p_2}dp_2}$, we obtain
		\[
		Q_2(u_2) = \frac{\left[(1 - u_2)^{-\frac{1}{k_1}}\right]}{\alpha_2}.
		\]
		Now setting $Q_2(u_2) = x_2$ and $u_2 = F_2 (x_2)$ and substituting in \eqref{413}, yield the required joint survival function of bivariate Pareto II in \eqref{411}.  The converse part of bivariate finite range distribution \eqref{412} is similar.  This completes the proof of the theorem.
	\end{proof}
\end{thm}

\section{Estimation of the Bivariate Excess Wealth Transform}

Let $\left\{(X_{1i},X_{2i}),\; i=1,2,\ldots,n\right\}$ be a random sample from the joint distribution of $(X_1, X_2)$.
The objective of this section is to construct empirical estimators of the bivariate excess wealth transforms $W_{12}$ and $W_{21}$.

\subsection{Weighted Empirical Estimation}

Suppose that $(X_{11}, X_{21}), (X_{12}, X_{22}), \cdots,(X_{1n}, X_{2n})$ is a random sample from the distribution of $(X_1, X_2)$, with positive sampling weights $\omega_1,\cdots,\omega_n$.  The weighted empirical distribution function of $Y$ is defined by
\begin{equation}
	\widehat F_{X_, n}(x) = \frac{\sum_{i=1}^n \omega_i (X_i \leq x)}{\sum_{i=1}^n\omega_i}, \; i = 1, 2
	\label{eq:weighted-empirical-FY}
\end{equation}
The corresponding weighted empirical quantile function is
\begin{equation}
	\widehat Q_{X_1, n}(u)
	=
	\inf \left\{
	x_1:
	\widehat F_{X_1,n}(x)\geq u_1
	\right\}, \; 0 < u_1 < 1.
	\label{eq:weighted-quantile-Y}
\end{equation}
Similarly,
\begin{equation}
	\widehat Q_{X_2, n}(u) = \inf \left\{x_2: \widehat F_{X_2, n}(x)\geq u_2 \right\}, \; 0 < u_2 < 1.
	\label{eq:weighted-quantile-X}
\end{equation}
See \citet{chambers1986estimating}, \cite{hall1999methods} and \cite{withers2010distribution} for details about the weighted empirical estimation.

\subsection{Conditional tail sample and estimator of $W_{12}$}

For a fixed $u_2\in(0,1)$, define the estimated threshold $\widehat t_{X_2, n}(u_2)
=
\widehat Q_{X_2,n}(u_2)$, with the estimated upper-tail set as $\widehat A_{2,n}(u_2)
=
\left\{
i:
X_{2i}>\widehat t_{X_2, n}(u_2)
\right\}$.  The use of the strict inequality in $\widehat A_{2,n}(u_2)$ means that observations exactly equal to the estimated threshold are excluded.  Let
\begin{equation}
	\widehat \Omega_{2,n}(u_2)
	=
	\sum_{i\in\widehat A_{2,n}(u_2)}
	\omega_i.
	\label{eq:Omega2}
\end{equation}
Then the normalized conditional weights are
\begin{equation}
	v_{i,n}^{(2)}
	=
	\frac{\omega_i}
	{\widehat\Omega_{2,n}(u_2)},
	\qquad
	i\in\widehat A_{2,n}(u_2).
	\label{eq:conditional-weights}
\end{equation}
The weighted empirical conditional distribution of $X_1$ given the upper-tail event of $X_2$ is therefore
\begin{equation}
	\widehat F_{12, n}(x_1 | Q_2(u_2)) = \frac{\sum_{i=1}^n \omega_i \{X_{2i} > \widehat Q_{X_2,n}(u_2)\} \{X_{1i}\leq x_1\}}{\sum_{i=1}^n \omega_i \{X_{2i} > \widehat Q_{X_2, n}(u_2)\}}.
	\label{eq:F12hat}
\end{equation}
The corresponding weighted empirical conditional quantile is
\begin{equation}
	\widehat Q_{12,n}(u_1, u_2)
	=
	\inf
	\left\{
	x_1:
	\widehat F_{12,n}\left(x_1 | Q_2 (u_2)\right) \geq u_1
	\right\}.
	\label{eq:Q12hat}
\end{equation}
Using \eqref{eq:Q12hat}, a natural estimator of $W_{12}$ then becomes
\begin{equation}
	\widehat W_{12,n}(u_1,u_2)
	=
	\frac{
		\displaystyle
		\sum_{i=1}^n
		\omega_i
		\mathbb{I}\left\{
		X_{2i} > \widehat Q_{X_2,n}(u_2)
		\right\}
		\left[
		X_{1i}-\widehat Q_{12,n}(u_1, u_2)
		\right]_+
	}{
		\displaystyle
		\sum_{i=1}^n
		\omega_i
		\mathbb{I}\left\{
		X_{2i} > \widehat Q_{X_2,n}(u_2)
		\right\}
	}.
	\label{eq:W12hat-positive}
\end{equation}
Equivalently,
\begin{equation}
	\widehat W_{12,n}(u_1,u_2)
	=
	\sum_{i\in\widehat A_{2,n}(u_2)}
	v_{i,n}^{(2)}
	\left[
	X_i-\widehat Q_{12,n}(u_1, u_2)
	\right]_+.
	\label{eq:W12hat-normalized}
\end{equation}

\subsection{Finite-sample summation representation}

Let the observations in the estimated conditional tail sample be
ordered as $X_{(1)}^{(2)}
\leq
X_{(2)}^{(2)}
\leq\cdots\leq
X_{(m)}^{(2)}$, where $m=|\widehat A_{2,n}(u_2)|$.  Let the corresponding normalized weights be $v_1, v_2, \cdots, v_m,$ and define $V_j=\sum_{\ell=1}^j v_\ell$.  Let $k = \min\{j:V_j\geq u_1\}$.  Then $\widehat Q_{12,n}(u_1\mid u_2) = X_{(k)}^{(2)}$.  The integral of the empirical conditional quantile function is
\begin{equation}
	\int_{u_1}^{1}
	\widehat Q_{12,n}(u_1, u_2)\,dp
	=
	X_{(k)}^{(2)}(V_k-u_1)
	+
	\sum_{j=k+1}^{m}
	X_{(j)}^{(2)}v_j.
	\label{eq:integral-empirical-quantile}
\end{equation}
Consequently,
\begin{equation}\label{eq:W12hat-sum}
	\widehat W_{12,n}(u_1,u_2)
	= \int_{u_1}^{1}
	\widehat Q_{12,n}(p_1, u_2)\,dp_1
	-
	(1-u_1)
	\widehat Q_{12,n}(u_1, u_2) =
	\sum_{j=k+1}^{m}
	\left(
	X_{(j)}^{(2)}
	-
	X_{(k)}^{(2)}
	\right)v_j.
\end{equation}
Thus, the empirical BEWT can be computed exactly from the weighted
conditional order statistics without numerical integration.

\begin{rem}
	Equation \eqref{eq:W12hat-sum} is equivalent to the positive-part
	representation in \eqref{eq:W12hat-positive}. The positive-part form is
	particularly convenient for implementation, whereas the order-statistic
	form makes the connection with the empirical integrated quantile
	representation explicit.\\
\end{rem}
The estimator of $W_{21}$ can be obtained in a similar manner, by defining the tail threshold, $\widehat t_{X_1, n}(u_1) = \widehat Q_{X_1, n}(u_1)$ with upper tail set, $\widehat A_{1,n}(u_1) = \left\{
i: X_{1i} > \widehat Q_{X_1, n}(u_1)\right\}$.  Then
\begin{equation}
	\widehat W_{21,n}(u_1, u_2)
	=
	\frac{
		\displaystyle
		\sum_{i=1}^n
		\omega_i
		\mathbb{I}\left\{
		X_{1i}>\widehat Q_{X_1, n}(u_1)
		\right\}
		\left[
		Y_i-\widehat Q_{21,n}(u_1, u_2)
		\right]_+
	}{
		\displaystyle
		\sum_{i=1}^n
		\omega_i
		\mathbb{I}\left\{
		X_{1i}>\widehat Q_{X_1, n}(u_1)
		\right\}
	}.
	\label{eq:W21hat}
\end{equation}
Equivalently,
\begin{equation}
	\widehat W_{21,n}(u_1,u_2)
	=
	\int_{u_2}^{1}
	\widehat Q_{21,n}(u_1, p_2)\,dp_2
	-
	(1-u_2)
	\widehat Q_{21,n}(u_1, u_2) =
	\sum_{j=k+1}^{m}
	\left(
	X_{(j)}^{(1)}
	-
	X_{(k)}^{(1)}
	\right)v_j.
	\label{eq:W21hat-integrated}
\end{equation}

\subsection{Consistency of the BEWT Estimators}

We now establish consistency of the proposed estimators. Let
$0<\epsilon<\frac12$ and consider quantile indices in the compact set $\mathcal U_\epsilon=[\epsilon,1-\epsilon]$.  The restriction away from zero and one avoids unstable marginal tail
thresholds and ensures that the conditioning events have probabilities
bounded away from zero.

\subsubsection{Assumptions}

\begin{ass}[Positive and regular sampling weights]
	\label{ass:weights}
	The sampling weights satisfy $0<c_1\leq \omega_i\leq c_2<\infty$ with probability tending to one, for some constants $c_1,c_2>0$.  Moreover, the weighted empirical averages satisfy a weighted law of
	large numbers.
\end{ass}

\begin{ass}[Marginal continuity]
	\label{ass:marginal}
	The marginal distribution functions $F_X$ and $F_Y$ are continuous
	and strictly increasing in neighborhoods of $\{Q_X(u):u\in\mathcal U_\epsilon\}$ and $\{Q_Y(u):u\in\mathcal U_\epsilon\}$.
\end{ass}

\begin{ass}[Conditional continuity]
	\label{ass:conditional}
	For every $u_1, u_2 \in \mathcal U_\epsilon$, the conditional distributions $F_{ij}\left(x_i | X_j > Q_j (u_j)\right), \; i, j = 1, 2, \; i \neq j$ are continuous and strictly increasing on the relevant support.
\end{ass}


\begin{ass}[Uniform integrability]
	\label{ass:moment}
	There exists $\delta>0$ such that 
	\[
	\sup_{u_2\in\mathcal U_\epsilon} \E\left[|X_1|^{1+\delta} | X_2 > Q_2(u_2)\right] <\infty \quad \text{and} \quad \sup_{u_1\in\mathcal U_\epsilon}
	E\left[|X_2|^{1+\delta} | X_1 > Q_1(u_1)\right] <\infty.
	\]
\end{ass}

\begin{ass}[Uniform weighted empirical-process consistency]
	\label{ass:empirical}
	The weighted empirical distribution functions satisfy 
	\[
	\sup_x \left| \widehat F_{X_1, n}(x)- F_{X_1}(x) \right| \xrightarrow{p} 0 \quad \text{and} \quad \sup_x \left| \widehat F_{X_2, n}(x) - F_{X_2}(x) \right| \xrightarrow{p} 0.
	\]
	Furthermore, the corresponding weighted empirical conditional distribution functions satisfy
	\[
	\sup_{\substack{x\in \R \\u_2\in\mathcal U_\epsilon}} \left| \widehat F_{12, n}\left(x_1 | Q_2 (u_2)\right)
	- F_{12}\left(x_1 | Q_2 (u_2)\right) \right| \xrightarrow{p} 0
	\]
	and 
	\[
	\sup_{\substack{y\in \R \\u_1\in\mathcal U_\epsilon}}
	\left|
	\widehat F_{21,n}\left(x_2 | Q_1 (u_1)\right) -
	F_{21}\left(x_2 | Q_1 (u_1)\right)
	\right| \xrightarrow{p}0.
	\]
\end{ass}

\begin{lem}
	Under Assumptions~\ref{ass:weights}--\ref{ass:empirical},
	\[
	\sup_{u_1\in\mathcal U_\epsilon}
	\left| \widehat Q_{X_1, n}(u_1) - Q_1(u_1)
	\right| \xrightarrow{p}0 \quad \text{and} \quad \sup_{u_2\in\mathcal U_\epsilon} \left|
	\widehat Q_{X_2, n}(u_2) - Q_2(u_2) \right| \xrightarrow{p}0
	\]
	Further, for every fixed $\eta>0$, 
	\begin{equation}
		\sup_{\substack{
				u_1 \in [\epsilon,1-\eta]\\
				u_2\in\mathcal U_\epsilon
		}}
		\left| \widehat Q_{12,n}(u_1, u_2) - Q_{12}(u_1, u_2) \right| \xrightarrow{p}0,
		\label{eq:Q12-uniform}
	\end{equation}
	and 
	\begin{equation}
		\sup_{\substack{u_2\in[\epsilon,1-\eta]\\
				u_1\in\mathcal U_\epsilon}}
		\left| \widehat Q_{21,n}(u_1, u_2) -
		Q_{21}(u_1, u_2) \right| \xrightarrow{p}0.
		\label{eq:Q21-uniform}
	\end{equation}
\end{lem}

\begin{proof}
	The result follows from the uniform convergence of the weighted empirical distribution functions and the continuity and strict monotonicity of the marginal distribution functions on the relevant quantile ranges. The inverse-map theorem for quantile functions further provides the uniform convergence of the empirical quantiles.
	
	The convergence of conditional quantile functions follows from the uniform convergence of the conditional empirical distribution functions and the continuity and strict
	monotonicity of the corresponding conditional distribution functions.  The restriction $p\leq1-\eta$ avoids the upper endpoint where quantile estimation can be unstable.
\end{proof}

For the BEWT, we establish convergence in the integrated sense.
\begin{thm}
	Under the assumptions~\ref{ass:weights}--\ref{ass:moment}, for every fixed $u_1 \in (0,1)$,
	\begin{equation}
		\sup_{u_2\in\mathcal U_\epsilon}
		\int_{u_1}^{1}
		\left|
		\widehat Q_{12, n}(p_1,  u_2) - Q_{12}(p_1, u_2)
		\right|\,dp_1
		\xrightarrow{p}0.
		\label{eq:integrated-Q12}
	\end{equation}
	and
	\begin{equation}
		\sup_{u_1\in\mathcal U_\epsilon}
		\int_{u_2}^{1}
		\left|
		\widehat Q_{21, n}(u_1, p_2) - Q_{21}(u_1, p_2)
		\right|\,dp_2
		\xrightarrow{p}0.
		\label{eq:integrated-Q21}
	\end{equation}
\end{thm}

\begin{proof}
	Fix $\eta>0$ and decompose
	\begin{align}
		\int_{u_1}^{1}
		\left|
		\widehat Q_{12,n}(p_1, u_2)
		-
		Q_{12}(p_1, u_2)
		\right|dp_1
		&=
		\int_{u_1}^{1-\eta}
		\left|
		\widehat Q_{12,n}(p_1, u_2)
		-
		Q_{12}(p_1,u_2)
		\right|dp_1
		\nonumber\\
		&\quad+
		\int_{1-\eta}^{1}
		\left|
		\widehat Q_{12, n}(p_1, u_2)
		-
		Q_{12}(p_1, u_2)
		\right|dp.
	\end{align}
	For the first integral, $\int_{u_1}^{1-\eta}
	\left| \widehat Q_{12, n}(p_1, u_2) - Q_{12}(p_1, u_2)
	\right|dp_1$ is bounded by 
	\[
	(1-u_1)
	\sup_{\substack{
			p\in[u_1,1-\eta]\\
			u_2\in\mathcal U_\epsilon
	}}
	\left|
	\widehat Q_{12,n}(u_1, u_2)
	-
	Q_{12}(u_1, u_2)
	\right|,
	\]
	which converges to zero in probability by
	\eqref{eq:Q12-uniform}.  For the upper tail interval $(1-\eta,1)$, Assumption~\ref{ass:moment}
	provides uniform integrability of the conditional distributions.  Consequently, the contribution from this interval can be made arbitrarily small by choosing $\eta$ sufficiently small.  Hence,
	\[
	\sup_{u_2\in\mathcal U_\epsilon}
	\int_{u_1}^{1}
	\left|
	\widehat Q_{12,n}(p_1, u_2)
	-
	Q_{12}(p_1, u_2)
	\right|dp_1
	\xrightarrow{p}0.
	\]
	The argument for $Q_{21}$ is identical.
\end{proof}

We can now establish the main consistency result.
\begin{thm}
	\label{thm:W12-consistency}
	Suppose that Assumptions~\ref{ass:weights}--\ref{ass:empirical}
	hold. 
	Then
	\begin{equation}
		\sup_{u_2\in\mathcal U_\epsilon}
		\left|
		\widehat W_{12,n}(u_1,u_2)
		-
		W_{12}(u_1,u_2)
		\right|
		\xrightarrow{p}0.
		\label{eq:W12-uniform-consistency}
	\end{equation}
	\begin{equation}
		\sup_{u_1\in\mathcal U_\epsilon}
		\left|
		\widehat W_{21,n}(u_1,u_2)
		-
		W_{21}(u_1,u_2)
		\right|
		\xrightarrow{p}0.
		\label{eq:W21-uniform-consistency}
	\end{equation}
\end{thm}

\begin{proof}
	Using the integrated-quantile representations,
	\begin{align}
		\left|
		\widehat W_{12,n}(u_1,u_2)
		-
		W_{12}(u_1,u_2)
		\right|
		&\leq 
		\left|
		\int_{u_1}^{1}
		\left[
		\widehat Q_{12, n}(p_1, u_2)
		-
		Q_{12}(p_1, u_2)
		\right]dp_1
		\right|
		\nonumber\\
		&\quad+
		(1-u_1)
		\left|
		\widehat Q_{12, n}(u_1, u_2)
		-
		Q_{12}(u_1, u_2)
		\right|
		\nonumber\\
		&\leq
		\int_{u_1}^{1}
		\left|
		\widehat Q_{12, n}(p_1, u_2)
		-
		Q_{12}(p_1, u_2)
		\right|dp_1
		\nonumber\\
		&\quad+
		(1-u_1)
		\left|
		\widehat Q_{12, n}(u_1, u_2)
		-
		Q_{12}(u_1, u_2)
		\right|.
	\end{align}
	Taking the supremum over $u_2\in\mathcal U_\epsilon$ gives
	\begin{align}
		\sup_{u_2\in\mathcal U_\epsilon}
		\left|
		\widehat W_{12,n}(u_1,u_2)
		-
		W_{12}(u_1,u_2)
		\right|
		&\leq 
		\sup_{u_2\in\mathcal U_\epsilon}
		\int_{u_1}^{1}
		\left|
		\widehat Q_{12,n}(p_1, u_2)
		-
		Q_{12}(p_1, u_2)
		\right|dp_1
		\nonumber\\
		&\quad+
		(1-u_1)
		\sup_{u_2\in\mathcal U_\epsilon}
		\left|
		\widehat Q_{12,n}(u_1, u_2)
		-
		Q_{12}(u_1, u_2)
		\right|.
	\end{align}
	
	Both terms converge to zero in probability. Hence
	\eqref{eq:W12-uniform-consistency} follows.
	The proof of \eqref{eq:W21-uniform-consistency} follows immediately by interchanging the roles of $X_1$ and
	$X_2$.
\end{proof}

\section{Simulation Study}
To examine the finite-sample behavior of the proposed nonparametric estimator of the new BEWT, a Monte Carlo simulation study was conducted under the bivariate Pareto model \eqref{Pareto II} introduced in Example \ref{ex22}. The distribution was chosen due to its ability to
capture the heavy-tailed characteristics commonly observed in income and wealth data. The estimator was evaluated in the probability combinations $(u_1, u_2)$ = (0.5, 0.5), (0.7, 0.7), (0.9, 0.9), (0.5, 0.7), (0.7, 0.9), and (0.5, 0.9) corresponding to lower, central, and upper regions of the conditional distribution. For each setting, 500 Monte Carlo replications were generated for sample sizes $n$ = 100, 250, 500, 1000 and 2000. The resulting estimates of $\widehat W_{12}(u_1,u_2)$ and 
$\widehat W_{21}(u_1,u_2)$ were
assessed using absolute bias and mean squared error (MSE). The simulation results are summarized in Table \ref{tab:simulation_W12}. It shows that the Monte Carlo means of
$\widehat{W}_{12}$ approach their corresponding theoretical values as
the sample size increases. In general, both the absolute bias and RMSE
decrease with increasing sample size, indicating improved finite-sample
performance of the estimator. This behavior is observed for both equal
and unequal quantile levels.  The simulation results provide numerical evidence supporting the consistency of the proposed empirical BEWT estimators.

\begin{table}[htbp]
	\small
	\centering
	\caption{Monte Carlo performance of the empirical BEWT estimator
		$\widehat W_{12}(u_1,u_2)$ and 
		$\widehat W_{21}(u_1,u_2)$ under the bivariate Pareto II model in \eqref{Pareto II} for different
		sample sizes and quantile levels.}
	\label{tab:simulation_W12}
	\begin{tabular}{cccccccc}
		\toprule
		$(u_1, u_2)$ & $n$ & Mean($(\widehat W_{12},\widehat W_{21}) $) & Bias &  RMSE \\
		\midrule
		(0.50, 0.50) & 100 & (0.2907, 0.2886) & (0.0091, 0.0070) &                   (0.0385, 0.0382) \\
		& 250 & (0.2836, 0.02853) & (0.0021, 0.0037) & (0.0220, 0.0226) \\
		& 500 & (0.2834, 0.2842) & (0.0018, 0.0026) & (0.0156, 0.0150) \\    
		& 1000 & (0.2824, 0.2821) & (0.0008, 0.0006) & (0.0113, 0.0114) \\
		& 2000 & (0.2826, 0.2826) & (0.0011, 0.0011) & (0.0075, 0.0074) \\
		% ...
		(0.70, 0.70) & 100 & (0.1769, 0.1780) & (0.0141, 0.0151) &                   (0.0513, 0.0519) \\
		& 250 & (0.1694, 0.1690) & (0.0066, 0.0061) & (0.0271, 0.0265) \\
		& 500 & (0.1664, 0.1653) & (0.0036, 0.0024) & (0.0177, 0.0184) \\
		& 1000 & (0.1648, 0.1649) & (0.0020, 0.0020) & (0.0126, 0.0129) \\
		& 2000 & (0.1633, 0.1630) & (0.0004, 0.0002) & (0.0093, 0.0081) \\
		% ...
		(0.90, 0.90) & 100 & (----,----) & (----,----) & (----,----) \\
		& 250 & (0.0703, 0.0719) & (0.0145, 0.0160) & (0.0619, 0.0624) \\
		& 500 & (0.0618, 0.0637) & (0.0060, 0.0078) & (0.0242, 0.0265) \\
		& 1000 & (0.0593, 0.0587) & (0.0034, 0.0028) & (0.0143, 0.0151) \\
		& 2000 & (0.0578, 0.0580) & (0.0020, 0.0021) & (0.0097, 0.0102) \\
		% ...
		(0.50, 0.70) & 100 & (0.3450, 0.1418) & (0.1785, 0.0056) &                   (0.0698, 0.0274) \\
		& 250 & (0.3365, 0.1393) & (0.0093, 0.0030) & (0.0398, 0.0151) \\
		& 500 & (0.3305, 0.1370) & (0.0034, 0.0007) & (0.0260, 0.0101) \\
		& 1000 & (0.3296, 0.1369) & (0.0025, 0.0006) & (0.0188, 0.0074) \\
		& 2000 & (0.3290, 0.1370) & (0.0019, 0.0007) & (0.0123, 0.0052) \\
		% ...
		(0.70, 0.90) & 100 & (----, 0.0425) & (----, 0.0081) & (---                  -, 0.0301) \\
		& 250 & (0.2803, 0.0381) & (0.0312, 0.0038) & (0.0979, 0.0123) \\
		& 500 & (0.2679, 0.0365) & (0.0188, 0.0021) & (0.0649, 0.0079) \\
		& 1000 & (0.2594, 0.0357) & (0.0103, 0.0014) & (0.0387, 0.0049) \\
		& 2000 & (0.2525, 0.0348) & (0.0034, 0.0005) & (0.0276, 0.0031) \\
		% ...
		(0.50, 0.90) & 100 & (----, 0.0323) & (----, 0.0045) & (---                  -,0.0131) \\
		& 250 & (0.5174, 0.0294) & (0.0423, 0.0017) & (0.1504, 0.0061) \\
		& 500 & (0.4902, 0.0282) & (0.0151, 0.0004) & (0.0817, 0.0037) \\
		& 1000 & (0.4864, 0.0280) & (0.0113, 0.0002) & (0.0602, 0.0026) \\
		& 2000 & (0.4809, 0.0280) & (0.0059, 0.0003) & (0.0396, 0.0019) \\
		\bottomrule
	\end{tabular}
	
	\begin{flushleft}
		\footnotesize
		Note: A dash (--) indicates that the empirical BEWT could not be
		computed because the conditional sample size was insufficient.
	\end{flushleft}
\end{table}

\section{Data Analysis}
To demonstrate the applicability of the proposed BEWT in the analysis of socioeconomic wealth
inequality, we consider microdata from the Eurosystem Household Finance and Consumption Survey (HFCS), coordinated by the European Central Bank (ECB) and the Household Finance and Consumption Network (HFCN). The HFCS is a harmonized household survey designed to collect detailed information on household finances and consumption, including assets, liabilities, income, private businesses, intergenerational transfers and saving behaviour. Its standardized core variables facilitate cross-country comparisons while allowing national institutions some flexibility in survey implementation.

The HFCS is particularly suitable for the present analysis because it provides household-level information on wealth and income, allowing the distribution and upper-tail behaviour of household economic resources to be examined directly. The survey is conducted at the national level under a common methodological framework, and anonymised microdata are made available to researchers. In this study, the HFCS data are used to illustrate the empirical behaviour of the proposed BEWT across selected European countries, with particular emphasis on the relationship between household wealth and income in the upper part of their respective distributions. The variable $X_1$ represents the net wealth and $X_2$ represents
the total household gross income.
The proposed nonparametric estimator was employed to estimate the excess wealth surfaces
$\widehat W_{12}(u_1,u_2)$ and 
$\widehat W_{21}(u_1,u_2)$ over a grid of probability values in the unit square. Selected
values of the estimated BEWT are reported in Table \ref{tab:w12-country} and \ref{tab:w21-country}.

\begin{table}[htbp]
	\centering
	\caption{Pooled $W_{12}$ estimates across countries (EUR million)}
	\label{tab:w12-country}
	\resizebox{\textwidth}{!}{%
		\begin{tabular}{lrrrrrr}
			\toprule
			Country & $(.50,.50)$ & $(.50,.70)$ & $(.50,.90)$ &
			$(.70,.70)$ & $(.70,.90)$ & $(.90,.90)$ \\
			\midrule
			Luxembourg   & 943752.5 & 1151902.9 &
			1919878.1 &
			877728.9 &
			1402277.4 &
			715066.7 \\
			Spain    & 335697.2 &
			427195.5 &
			693264.5 &
			334883.2 &
			551490.5 &
			403835.6 \\
			Germany     & 326003.3 &
			384882.9 &
			591906.9 &
			289352.5 &
			413931.5 &
			205495.6 \\
			Belgium     & 320882.6 &
			387300.6 &
			517109.7 &
			280205.4 &
			374160.1 &
			229959.9 \\
			France     & 317654.1 &
			402927.7 &
			623231.9 &
			295899.9 &
			450672.4 &
			255105.4 \\
			Italy  & 300489.3 &
			409956.8 &
			834456.8 &
			335905.8 &
			725787.9 &
			511203.3 \\
			\bottomrule
		\end{tabular}%
	}
\end{table}

Table \ref{tab:w12-country} reports the pooled estimates of \(W_{12}\) for the six selected European countries. The estimates exhibit substantial cross-country heterogeneity in the conditional excess net wealth structure. Luxembourg records considerably larger  excess net wealth values \((W_{12})\) than the other countries in all reported quantile pairs. For example, at \((u_1,u_2)=(0.50,0.50)\), the estimated \(W_{12}\) for Luxembourg is approximately €0.94 million, compared with approximately €0.32--€0.34 million for Belgium, Germany, Spain and France and €0.30 million for Italy. 

A clear pattern is also observed with respect to the conditioning quantile \(u_2\), as obtained from Table \ref{tab:w12-country}. Keeping \(u_1=0.50\) fixed, the \(W_{12}\) estimates increase as \(u_2\) moves from 0.50 to 0.90 for every country. For Luxembourg, for example, the estimate increases from approximately €0.94 million at \((0.50,0.50)\) to approximately €1.92 million at \((0.50,0.90)\). The corresponding increase is also evident for the other countries, although the magnitude differs substantially across countries.  This pattern indicates that the excess net wealth component measured by \(W_{12}\) becomes larger when the conditioning group is restricted to households above higher income quantiles.   The country comparison suggests that the conditional excess-wealth structure varies more strongly across countries when attention is concentrated on households above higher income thresholds. 


\begin{table}[!t]
	\centering
	\caption{Pooled $W_{21}$ estimates across countries (EUR million)}
	\label{tab:w21-country}
	\resizebox{\textwidth}{!}{%
		\begin{tabular}{lrrrrrr}
			\toprule
			Country & $(.50,.50)$ & $(.50,.70)$ & $(.50,.90)$ &
			$(.70,.70)$ & $(.70,.90)$ & $(.90,.90)$ \\
			\midrule
			Luxembourg    & 55713.51 &
			34453.88 &
			14399.084 &
			40746.35 &
			17594.258 &
			24551.542 \\
			Spain   & 18436.68 &
			11184.43 &
			4751.104 &
			13557.59 &
			5307.188 &
			9985.235 \\
			Germany       & 42597.67 &
			31289.72 &
			17797.918 &
			38272.87 &
			22692.216 &
			24232.410 \\
			Belgium     & 27820.34 &
			16367.67 &
			6244.871 &
			19731.52 &
			6530.747 &
			6612.846 \\
			France    & 24832.90 &
			18005.15 &
			10115.016 &
			23363.78 &
			13937.873 &
			26505.755 \\
			Italy & 23058.19 &
			15763.98 &
			8738.590 &
			20568.33 &
			11129.126 &
			18403.766 \\
			\bottomrule
		\end{tabular}%
	}
\end{table}
Table 3 reports the pooled estimates of the bivariate excess wealth transform \(W_{21}\) across the six countries. The estimates exhibit considerable cross-country heterogeneity in the conditional excess household income structure. At the \((0.50,0.50)\) quantile pair, Luxembourg records the largest value of \(W_{21}\), at approximately €55.71 million, followed by Germany at €42.60 million. Spain, Belgium, France and Italy exhibit comparatively smaller values, ranging from approximately €18.44 million to €27.82 million.

A clear pattern emerges when the wealth-conditioning quantile is held fixed at \(u_1=0.50\). As the income quantile \(u_2\) increases from 0.50 to 0.70 and then to 0.90, \(W_{21}(0.50,u_2)\) decreases in all six countries. For Luxembourg, for example, the estimate declines from €55.71 million at \((0.50,0.50)\) to €34.45 million at \((0.50,0.70)\) and €14.40 million at \((0.50,0.90)\). A similar decline is observed for Germany, where the corresponding estimates decrease from €42.60 million to €31.29 million and €17.80 million.  This pattern indicates that, when the wealth-conditioning threshold is fixed, the excess-income component decreases as progressively higher income thresholds are considered. The magnitude of this decline differs across countries, suggesting heterogeneity in the joint distribution of household wealth and income. In particular, Luxembourg maintains comparatively high \(W_{21}\) values over much of the considered range, whereas Spain generally records the smallest values.  Now Figures \ref{fig2} and \ref{fig3} provide an overall comparison of the components of BEWT ($W_{12}$ and $W_{21}$) of household net wealth and income, conditioned by different quartiles of household income and net wealth respectively. 

The comparison of \(W_{12}\) and \(W_{21}\) highlights the directional nature of the proposed BEWT. The \(W_{12}\) transform measures conditional excess wealth under an income-tail conditioning event, whereas \(W_{21}\) measures conditional excess income under a wealth-tail conditioning event. The substantially different magnitudes and cross-country patterns of the two transforms indicate that the conditional relationship between wealth and income is not symmetric. Hence, the pair \((W_{12},W_{21})\) provides information about the joint distribution that cannot be obtained by considering the marginal distributions of wealth and income separately. 
\begin{figure}
	\centering    \includegraphics[width=0.60\linewidth]{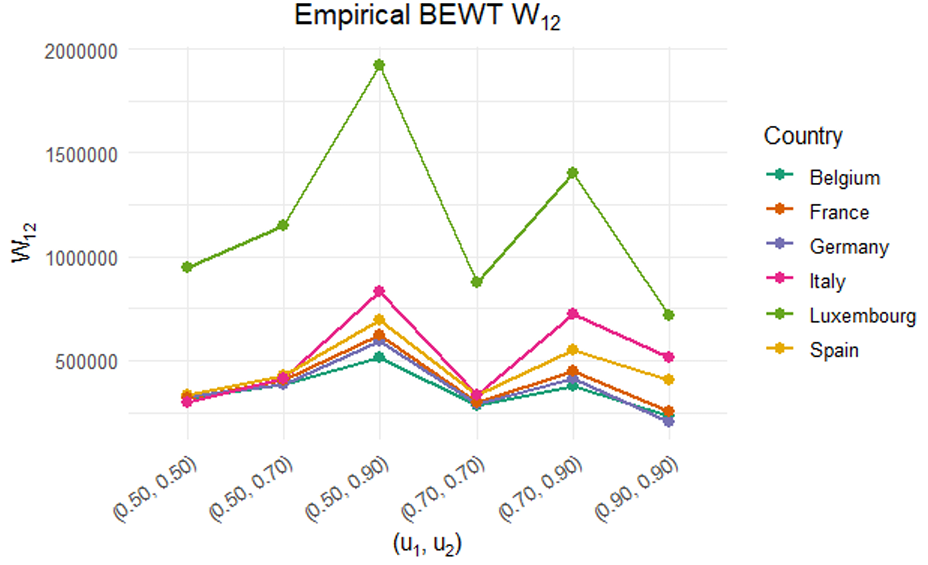}
	\caption{Empirical $W_{12}(u_1, u_2 )$ curves across the six countries}
	\label{fig2}
\end{figure}

\begin{figure}
	\centering    \includegraphics[width=0.6\linewidth]{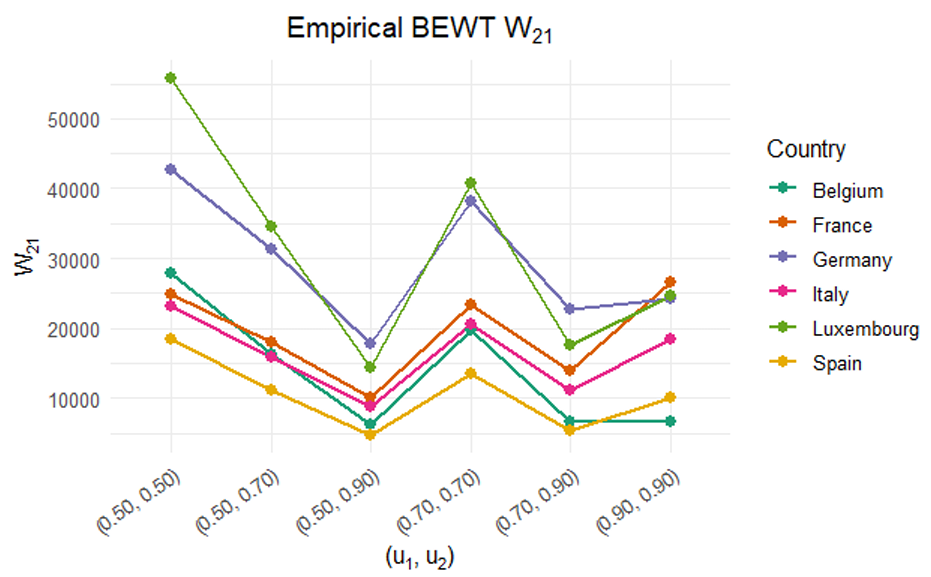}
	\caption{Empirical $W_{21}(0.5, u_2)$ curves across the six countries}
	\label{fig3}
\end{figure}

\section*{Data Availability Statement}
Data supporting all reported findings in this study are accessed from the European Central Bank (ECB), Eurosystem, by registration in https://astra.ecb.europa.eu/files.
\section*{Research ethics statement}
The data used in this study is obtained from the European Central Bank (ECB), Eurosystem, by registration in https://astra.ecb.europa.eu/files, and ethical approval was not required.

\section*{Acknowledgements}

The second author wish to thank the Cochin University of Science and Technology, India for the financial support.

\bibliographystyle{apalike}
\bibliography{references}
\end{document}